\documentclass[sigconf]{acmart}

\acmConference[ ]{ }{ }{ }
\setcopyright{none}
\acmPrice{}
\acmISBN{}
\acmDOI{}
\renewcommand\footnotetextcopyrightpermission[1]{} 
\usepackage{amsmath}
\usepackage{amsfonts}
\usepackage{graphicx}
\usepackage{multirow}
\usepackage{subfig}
\usepackage{pgf}

\def\eg{\textit{e.g.}}
\def\ie{\textit{i.e.}}

\def\md{\mathbf{d}}
\def\dq{\Delta Q}
\def\Prob{\mathbb{P}}
\def\R{\mathbb{R}}
\def\ed{($\epsilon$,$\delta$)}

\DeclareMathOperator{\Lap}{Lap}

\DeclareMathOperator{\h}{h}
\DeclareMathOperator{\range}{Range}


\newif\iffullpaper
\fullpapertrue

\begin{document}

\title{The Bounded Laplace Mechanism in Differential Privacy}

	\author{Naoise Holohan}
	\orcid{0000-0003-2222-9394}
	\affiliation{%
	\institution{IBM Research -- Ireland}
	}
	\email{naoise@ibm.com}
		 \author{Spiros Antonatos}
	 \affiliation{%
	\institution{IBM Research -- Ireland}
	}
	\email{santonat@ie.ibm.com}

	  \author{Stefano Braghin}
	  \affiliation{%
	\institution{IBM Research -- Ireland}
	}
	\email{stefanob@ie.ibm.com}

	 \author{P\'{o}l Mac Aonghusa}
	  \affiliation{%
	\institution{IBM Research -- Ireland}
	}
	\email{aonghusa@ie.ibm.com}

\begin{abstract}
The Laplace mechanism is the workhorse of differential privacy, applied to many instances where numerical data is processed.
However, the Laplace mechanism can return semantically impossible values, such as negative counts, due to its infinite support.
There are two popular solutions to this: (i) bounding\slash capping the output values and (ii) bounding the mechanism support.
In this paper, we show that bounding the mechanism support, while using the parameters of the pure Laplace mechanism, does not typically preserve differential privacy.
We also present a robust method to compute the optimal mechanism parameters to achieve differential privacy in such a setting.
\end{abstract}

\maketitle

\section{Introduction} \label{sc:intro}

Data privacy is an important factor that data owners must take into consideration when collecting, storing and publishing user data.
This extends to publishing statistics on user data.
In recent years, differential privacy has emerged as a popular privacy framework, thanks to its robust mathematical privacy guarantees.

The Laplace mechanism is the workhorse of differential privacy, frequently utilised in applications on numerical data.
Its strength lies in its mathematical and computational simplicity, in contrast to other mechanisms such as the exponential mechanism.
In spite of its popularity however, the Laplace mechanism lacks \emph{consistency} in its output.
Consider, for example, adding noise from the Laplace mechanism to a count query; negative results hold no meaning, yet are a valid output of the mechanism, occurring especially frequently for low-numbered counts.

\begin{example}\label{eg:1}
Suppose we are querying a census dataset, and seeking to learn the number of people born on Mars.
Adding noise from a Laplace mechanism with variance $\frac{2}{\epsilon^2}$ will satisfy differential privacy.
Although the real answer to the query is $0$ (for now at least!), we must add noise to protect the privacy of future human Martians.
Successive outputs from the Laplace mechanism could be: $-1.71$, $2.31$, $-1.20$, $0.652$.

However bizarre the query, negative outputs are patently illogical and inconsistent.
By the symmetry of the Laplace distribution, on average $50\%$ of the outputs will be negative.
\end{example}

Currently there are two solutions to this drawback, both involving the selection of an appropriate output domain.
If selection of the domain is done independently of the data, no privacy budget is consumed.
The first, \emph{truncation}, is to project values outside the domain to the closest value within the domain.
The second, \emph{bounding}, is to continue to sample independently from the mechanism until a value within the domain is returned.

\begin{example}
Using the same set-up as Example~\ref{eg:1}, if the Laplace mechanism returns a value $-1.71$, the truncation method projects the output to $0$ (the lower bound of a count query).
If the bounding method is used, the value is simply re-sampled, meaning the second value $2.31$ is returned (an analyst may subsequently wish to round this to $2$).
\end{example}

By design, the truncated Laplace mechanism has a (possibly large) non-zero probability of returning values at the domain bounds.
There are instances where this may be undesirable and\slash or incompatible, such as when the domain bounds coincide with singularities or values that otherwise result in a qualitative change in behaviour (\eg\ bifurcation points).
In such cases truncation may not be best-suited.

\begin{example}
Consider the case of releasing the variance of a distribution while using the Laplace mechanism to achieve differential privacy.
Zero variance is qualitatively different to non-zero variance, and may result in complications in its use.
In this case the bounded mechanism is a more appropriate choice as it has a zero probability of returning a zero variance.
\end{example}

In this paper we show that the bounded Laplace mechanism \emph{does not} typically satisfy differential privacy when inheriting parameters from the pure Laplace mechanism (see Section~\ref{sc:boundeddomain}).
In fact, in almost all cases, the variance of the Laplace distribution must be increased for the bounded Laplace mechanism to satisfy the same differential privacy constraints.

The statistical properties of the truncated and bounded Laplace mechanisms were initially studied in~\cite{Liu17}; further comparisons of the two mechanisms are beyond the scope of this paper.

Complete proofs to most lemmas and theorems are given in the Appendix.

\section{Preliminaries}

We first detail the notation that we'll use in this paper, broadly following the style introduced in~\cite{HLM15}.

We are interested in queries $Q: \mathcal{S}^n \to D$ on databases $\md \in \mathcal{S}^n$ mapping to a finite domain $D = [l, u]$ ($l < u$, both finite).
The sensitivity of $Q$ is defined in the usual way, $\Delta Q = \max_{\h(\md, \md^\prime) = 1} |Q(\md) - Q(\md^\prime)|$, where $\h: \mathcal{S}^n \times \mathcal{S}^n \to \mathbb{N}$ denotes Hamming distance.

In this paper we are only concerned with \emph{output perturbation} mechanisms, so we need only consider response mechanisms of the form $Y_q: \Omega \to \R$ for each $q \in D$ (since $Q(\mathcal{S}^n) \subseteq D$).
Given $\epsilon \ge 0$ and $0 \le \delta \le 1$, the mechanism $\{Y_q \mid q \in D\}$ satisfies {\ed}-differential privacy when
$$\Prob(Y_q \in A) \le e^\epsilon \Prob(Y_{q^\prime} \in A) + \delta,$$
for all measurable $A \subseteq \R$ and whenever $|q - q^\prime| \le \Delta Q$.

We denote by $\Lap(\mu, b)$ a Laplace distribution with mean $\mu$ and variance $2b^2$.
The standard Laplace mechanism is therefore given by
\begin{equation}\label{eq:laplacemech}
Y_{q} = q + \Lap(0, b) = \Lap(q, b),
\end{equation}
and satisfies {\ed}-differential privacy when $b \ge \frac{\dq}{\epsilon - \log(1-\delta)}$~\cite[Example~5]{HLM15}.
Note that $\Prob(Y_q \in \R \setminus D) > 0$, whereas we seek $\range(Y_q) = D$ for consistency.

\section{Bounded Laplace Mechanism} \label{sc:boundeddomain}

As the support of the Laplace distribution is infinite, it is common for the output of the Laplace mechanism to fall outside the range of $Q$.
Currently, there are two popular solutions to overcome this.
The first, which we will call \emph{truncation}, involves a deterministic mapping to the upper\slash lower bounds of the output domain, when the value falls outside.

Another approach is to bound the support of the response mechanism, and then sample directly from the output domain (\eg\ by inverse transform sampling).
This can also be achieved through rejection sampling, by continually redrawing from the unbounded distribution until an output falls within the domain.
We will refer to this process as \emph{bounding}, as the pure outputs of the mechanism are bounded by design.

\begin{definition}[Bounded Laplace Mechanism]\label{df:boundedlaplace}
Given $b > 0$ and $D \subset \R$, the \emph{bounded Laplace mechanism} $W_q: \Omega \to D$, for each $q \in D$, is given by its probability density function $f_{W_q}$:
$$f_{W_q}(x) = \begin{cases}
0, & \text{if } x \notin D, \\
\frac{1}{C_q} \frac{1}{2b} e^{-\frac{|x - q|}{b}}, & \text{if } x \in D,
\end{cases}$$
where $C_q = \int_D \frac{1}{2b} e^{-\frac{|x - q|}{b}} dx$ is a normalisation constant.
\end{definition}

\textbf{Remark 1:} It follows that $\Prob(W_q \in D) = 1$, and, conversely, that $\Prob(W_q \in \R \setminus D) = 0$.

\textbf{Remark 2:} Given $A \subseteq \R$, $\Prob(W_q \in A) = \frac{1}{C_q} \Prob(Y_q \in A \cap D)$, where $Y_q$ is given in (\ref{eq:laplacemech}).

As the output distribution is now a function of the query answer $Q(\md) = q$, the normalisation factor $C_q$ is no longer constant.
It is therefore no longer guaranteed that the mechanism $W_q$ satisfies differential privacy using parameters from the (pure) Laplace mechanism.

\subsection{Preliminary Results}

We first establish an algebraic representation for $C_q$.

\begin{lemma}\label{lm:cq}
For $C_q$ as given in Definition~\ref{df:boundedlaplace}, and for $q \in D = [l, u]$,
$$C_q = 1 - \frac{1}{2}\left(e^{-\frac{q - l}{b}} + e^{-\frac{u - q}{b}}\right).$$
\end{lemma}

We next consider the following lemma concerning $C_q$.

\begin{lemma}\label{lm:min}
Let $C_q$ be given by Definition~\ref{df:boundedlaplace}. Then,
$$\max_{\substack{q, q^\prime \in D \\ |q^\prime-q| \le \dq}} \frac{C_{q^\prime}}{C_q} e^{\frac{|q^\prime-q|}{b}} = \frac{C_{l+\dq}}{C_l} e^{\frac{\dq}{b}}.$$
\end{lemma}

\begin{proof}
	The following is an outline of the full proof given in Section~\ref{sc:a:lm:min}.
	By the symmetry of $C_q$ about $\frac{u+l}{2}$, we can assume that $q^\prime \ge q$.
	Showing that $\frac{\partial}{\partial z} \left( \frac{C_{q+z}}{C_q} e^{\frac{z}{b}}\right) \ge 0$ and $\frac{\partial}{\partial q} \left( \frac{C_{q+z}}{C_q} e^{\frac{z}{b}}\right) \le 0$ completes the proof.
\end{proof}

This leads us to the following definition of $\Delta C(b)$ for later use.

\begin{definition}
Given $C_q$ from Definition~\ref{df:boundedlaplace}, and noting that $C_q = C_q(b)$ is a function of $b$, we define $\Delta C(b)$ as follows:
$$\Delta C(b) = \frac{C_{l+\dq}(b)}{C_l(b)}.$$
\end{definition}

\subsection{Main Result}\label{sc:bounded:main}

We now proceed to the main result of this paper, which defines the variance required for the bounded Laplace mechanism.

\begin{theorem}\label{th:mainres}
Let $W_q$ be the bounded Laplace mechanism given in Definition~\ref{df:boundedlaplace} and let $\epsilon \ge 0$ and $0 \le \delta \le 1$ be given.
Then $\{W_q \mid q \in D\}$ satisfies {\ed}-differential privacy whenever
\begin{equation}\label{eq:mainres}
b \ge \frac{\dq}{\epsilon - \log \Delta C(b) - \log(1 - \delta)}.
\end{equation}
\end{theorem}

\begin{proof}
	The following is an outline of the full proof given in Section~\ref{sc:a:th:mainres}.
	We are seeking to show that
	$$\Prob(W_q \in A) \le e^\epsilon \Prob(W_{q^\prime} \in A) + \delta,$$
	for any measurable $A \subseteq D$ and where $q, q^\prime \in D$, $|q-q^\prime| \le \dq$.
	For this to hold it is sufficient to show that $1 \le e^{\epsilon - \frac{|q^\prime-q|}{b}} \frac{C_d}{C_{q^\prime}} + \delta$.
	Furthermore by Lemma~\ref{lm:min}, it is sufficient to show that
	$$1 \le \frac{1}{\Delta C(b)} e^{\epsilon-\frac{\dq}{b}} + \delta,$$
	which can be solved implicitly for $b$ to complete the proof.
\end{proof}

\begin{figure*}[t]
	\input{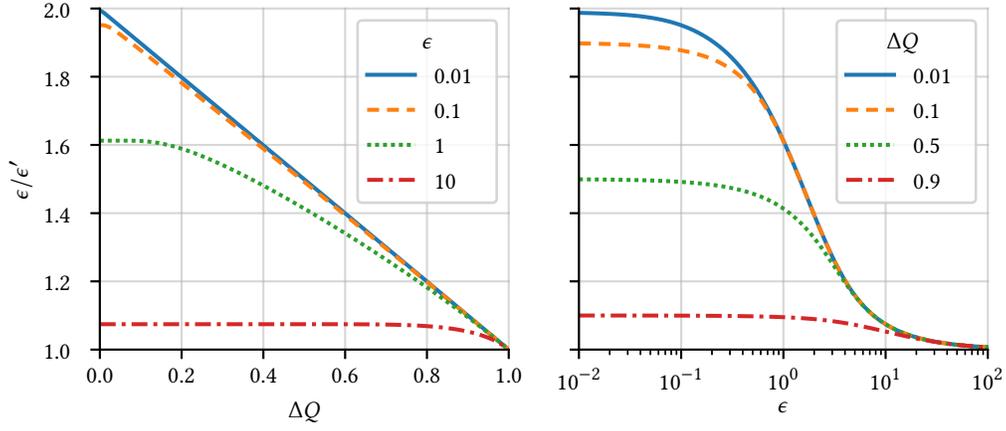}
  \caption{\label{fig:bounded:1} Relationship of $\frac{\epsilon}{\epsilon^\prime}$ to $\dq$ and $\epsilon$, where $u-l=1$ and $\delta=0$ are
 fixed.}
\end{figure*}

\textbf{Discussion:}
To satisfy {\ed}-differential privacy using the bounded Laplace mechanism, its variance will never be less than that of the (pure) Laplace mechanism (since $\Delta C(b) \ge 1$).
In the case of achieving $\epsilon$-differential privacy (\ie\ $\delta=0$), the underlying Laplace distribution must be one which satisfies $\epsilon^\prime$-differential privacy, where $\epsilon^\prime = \epsilon - \log \Delta C(b)$ (\ie\ for a target $\epsilon$, we require an effective $\epsilon^\prime$).
Inverse transform sampling or rejection sampling can then be used to determine the output.
As shown in Figure~\ref{fig:bounded:1}, the impact on $\epsilon^\prime$ is most pronounced when $\dq$ and $\epsilon$ are small; the graphical evidence aligns with the intuition that $2 \epsilon^\prime = \epsilon$ in the limiting case.

However, finding the optimal value for $b$ is non-trivial since the relationship given in Theorem~\ref{th:mainres} is implicitly defined.
This problem is examined in Section~\ref{sc:fixedpoint}.

The simpler task of determining a value of $\epsilon$ (or a relationship between $\epsilon$ and $\delta$) to a given value of $b$ can be achieved with (\ref{eq:mainres}).

\section{Calculating \lowercase{\texorpdfstring{$b$}{b}}} \label{sc:fixedpoint}

From the conclusion of Theorem~\ref{th:mainres}, let's define the following fixed point operator for $b$.

\begin{definition}[Fixed Point Operator]
Given $\dq > 0$, $\epsilon \ge 0$ and $0 \le \delta \le 1$, we define the fixed point operator $f: \R_{>0} \to \R_{>0}$ by
\begin{equation}\label{eq:fb}
f(b) = \frac{\dq}{\epsilon - \log \Delta C(b) - \log(1-\delta)}.
\end{equation}
\end{definition}

Any positive fixed point of $f$ (\ie\ $b^* = f(b^*) > 0$) will act as a differentially private shape parameter for the bounded Laplace mechanism.
In advance of examining $f$, we first define
$$b_0 = \frac{\dq}{\epsilon - \log(1-\delta)}.$$
Note that $b_0$ determines the variance required for the (pure) Laplace mechanism to achieve {\ed}-differential privacy.

We now present a number of lemmas concerning $f$, namely (i)~the value of $f(b_0)$ and (ii)~the monotonicity of $f$.
Proofs are given in Sections~\ref{sc:a:lm:posinitpoint} and \ref{sc:a:lm:fprime}.

\begin{lemma}\label{lm:posinitpoint}
$f(b_0) \ge b_0$,
and $f(b_0) = b_0$ if and only if $\dq = u-l$.
\end{lemma}

\begin{lemma}\label{lm:fprime}
$f^\prime(b) \le 0$ whenever $b \ne 0$, and $f^\prime(b) = 0$ if and only if $\dq = u-l$.
\end{lemma}

This leads us to the main result of this section, that $f$ has a unique fixed point $b^*$.

\begin{theorem}[Fixed Point]\label{th:existence}
There exists a unique $b^* \in [b_0, f(b_0)]$ such that $b^* = f(b^*)$, and $b^* = b_0 = f(b_0)$ if and only if $\dq = u-l$.
\end{theorem}

\begin{proof}
Since $f(b_0) \ge b_0$ (Lemma~\ref{lm:posinitpoint}), $f^\prime \le 0$ (Lemma~\ref{lm:fprime}) and $f(b)$ is continuous on $b \in [b_0, \infty)$ (since it is differentiable), it follows that $f(b)$ has a unique fixed point $b^* \in [b_0, \infty)$, where uniqueness follows from the monotonicity of $f$.

Furthermore, since $f^\prime \le 0$ and $f(b_0) \ge b_0$, it follows that $f(f(b_0)) \le f(b_0)$.
We must therefore have $b^* \in [b_0, f(b_0)]$.
And, since $f(b_0) = b_0$ if and only if $\dq = u-l$ (Lemma~\ref{lm:posinitpoint}), the result follows.
\end{proof}

It follows from Theorem~\ref{th:existence} that the mechanism $W_q$ from Definition~\ref{df:boundedlaplace} satisfies differential privacy for $b = b^*$.
Given that we have a bounded domain in which $b^*$ lies, and since $f$ is continuous, the bisection method is guaranteed to converge to $b^*$ for any given $\epsilon \ge 0$, $0 \le \delta \le 1$, $u > l$ and  $\dq \le u - l$.

\begin{theorem}
Let $b^* \in \R_{> 0}$ such that $b^* = f(b^*)$.
Then, given any $\xi > 0$,
$$b^* + \xi > f(b^* + \xi).$$
\end{theorem}

\begin{proof}
By Theorem~\ref{th:existence}, such a fixed point $b^*$ exists.
Furthermore, from Lemma~\ref{lm:fprime} we have $f^\prime(b) < 0$, hence
\begin{align*}
f(b^* + \xi) < f(b^*) &= b^* < b^* + \xi.\qedhere
\end{align*}
\end{proof}

Consequently by Theorem~\ref{th:mainres}, any fixed point $b^*$ is a lower bound on all values $b$ that satisfy {\ed}-differential privacy.

\section{Related Work}\label{sc:related}

In~\cite{Liu17}, the statistical properties of bounding and truncating the Laplace mechanism were explored, without examining the differential privacy properties of the bounded Laplace mechanism.
The same author followed with a study on generalised Gaussian mechanisms for differential privacy~\cite{Liu18}.
The results applied to the bounded Laplace mechanism showed a doubling of the noise variance ($\epsilon = 2 \epsilon^\prime$) is required, an increase we now know to be excessive.

In~\cite{ZZX12}, regression analysis under differential privacy was studied.
The authors looked to add noise (using the Laplace mechanism) to the coefficients of an objective function to achieve differential privacy, but this can result in an unbounded objective function.
Their first approach at solving this was to re-run the differential privacy mechanism until the result gives a solution to the optimisation problem.
This approach has the effect of doubling the noise variance (since $\epsilon = 2 \epsilon^\prime$), which our work has shown may be excessive.
The authors also proposed an alternative approach to maintain the privacy budget at $\epsilon$.

A na\"ive Bayes machine learning classifier was described in~\cite{VSB13}, which achieves differential privacy by adding Laplace noise to the model parameters.
For numerical data, na\"ive Bayes calculates the mean and standard deviation of the feature in order to classify unseen data.
The authors propose re-sampling from the Laplace distribution to ensure the differentially private standard deviations are positive, without modifying the variance.
From what we now know, this approach does not satisfy differential privacy.

Consistency in differential privacy has also been studied previously.
Examples include achieving consistent releases of margin\-als~\cite{BCD07} and histograms~\cite{HRM10}.
In~\cite{BCD07} the authors sought to release marginals consisting of non-negative integers, with consistent sums across marginals.
This was achieved using Fourier transformations and linear programming.
In~\cite{HRM10}, the authors used \emph{constrained inference} to ensure consistency in histogram counts through post-processing.

\section{Conclusion}

In this paper, we have shown that the bounded Laplace mechanism does not typically satisfy differential privacy when inheriting parameters from the Laplace mechanism, except in the case when $\dq = u-l$.
We have also presented details of calculating the required parameters for the bounded Laplace mechanism to satisfy differential privacy.
It was shown that the noise added to achieve differential privacy must be of greater variance than that of the pure Laplace mechanism.

The results of this paper highlight the dangers of re-sampling from the Laplace mechanism in applications of differential privacy to achieve valid\slash plausible outputs.
Researchers may be inadvertently violating differential privacy in doing so, or overcompensating by increasing the privacy budget excessively.
Our robust method of calculating the optimal noise variance will allow privacy researchers and practitioners to deploy the bounded Laplace mechanism with confidence and certainty.

\begin{acks}
The authors would like to thank the anonymous reviewers for their time in reading the paper and the helpful comments they provided.
\end{acks}

\pagebreak

\appendix

\section*{Appendix}
\setcounter{section}{1}

\subsection{Proof of Lemma~\ref{lm:min}}\label{sc:a:lm:min}

In order to prove Lemma~\ref{lm:min}, we must first consider the following lemmas concerning $C_q$.
	
\begin{lemma}\label{lm:dcdz}
Let $q \in D$ and $b>0$, and let $C_q$ be given by Definition~\ref{df:boundedlaplace}. Then
$\frac{\partial}{\partial z} \left( \frac{C_{q+z}}{C_q} e^{\frac{z}{b}}\right) \ge 0$,
whenever $q+z \le u$.
\end{lemma}

\begin{proof}
We first note that
$$\frac{\partial}{\partial z} C_{q+z} = \frac{1}{2b}\left(e^{-\frac{q+z-l}{b}} - e^{-\frac{u-q-z}{b}}\right).$$
We then see that
$$\frac{\partial}{\partial z} \left( \frac{C_{q+z}}{C_q} e^{\frac{z}{b}}\right) = \frac{1}{C_q} \frac{1}{b} \left( 1-e^{-\frac{u-q-z}{b}}\right) e^{\frac{z}{b}}.$$

Since $b>0$ by assumption, it follows that $\frac{\partial}{\partial z} \left( \frac{C_{q+z}}{C_q} e^{\frac{z}{b}}\right) \ge 0$ if and only if
$q+z \le u$.
\end{proof}
	
\begin{lemma}\label{lm:dcdq}
Let $q \in D$ and $z \ge 0$, and let $C_q$ be given by Definition~\ref{df:boundedlaplace}. Then
$\frac{\partial}{\partial q} \left( \frac{C_{q+z}}{C_q} e^{\frac{z}{b}}\right) \le 0$.
\end{lemma}

\begin{proof}
We first note that
$$\frac{\partial}{\partial q} C_{q+z} = \frac{1}{2b}\left(e^{-\frac{q+z-l}{b}} - e^{-\frac{u-q-z}{b}}\right).$$

We then find
\begin{align*}
\frac{\partial}{\partial q} \left(\frac{C_{q+z}}{C_q} e^{\frac{z}{b}} \right) &= \frac{e^{\frac{z}{b}}}{{C_q}^2} \left(C_q \frac{\partial}{\partial q} C_{q+z} - C_{q+z} \frac{\partial}{\partial q} C_q\right)\\
&= \frac{e^{\frac{z}{b}}}{2b\,{C_q}^2} \left( e^{-\frac{q-l}{b}} \left(e^{-\frac{z}{b}} - 1\right) + e^{-\frac{u-q}{b}} \left( 1-e^{\frac{z}{b}}\right)\right.\\
&\qquad\qquad \left. +\: e^{-\frac{u-l-z}{b}} - e^{-\frac{u-l+z}{b}}\right)\\
&= \frac{e^{\frac{z}{b}} \left( \left(e^{-\frac{z}{b}} - 1\right) \left(e^{\frac{u-q}{b}} - 1\right) + \left(1-e^{\frac{z}{b}}\right) \left( e^{\frac{q-l}{b}} - 1\right)\right)}{2b\, e^{\frac{u-l}{b}} {C_q}^2}.
\end{align*}

Since $b>0$, it's clear that the denominator is positive.
Furthermore, since $q \in D$, it follows that $e^{\frac{u-q}{b}}, e^{\frac{q-l}{b}} > 1$.
Also, since $z \ge 0$ by assumption, we have $e^{-\frac{z}{b}} < 1$ and $e^{\frac{z}{b}} > 1$.
Hence, $\frac{\partial}{\partial q} \left(\frac{C_{q+z}}{C_q} e^{\frac{z}{b}} \right) \le 0$, as required.
\end{proof}

Using Lemmas~\ref{lm:dcdz} and~\ref{lm:dcdq}, we can now prove Lemma~\ref{lm:min}.

\begin{proof}[Proof (Lemma~\ref{lm:min})]
Since $C_q$ is symmetric about $\frac{u+l}{2}$, we have $C_q = C_{u+l-q}$.
By letting $q_0 = u+l-q$ and $q_0^\prime = u+l-q^\prime$, then,
$\frac{C_{q^\prime}}{C_q} e^{\frac{|q^\prime-q|}{b}} = \frac{C_{q_0^\prime}}{C_{q_0}} e^{\frac{|q_0^\prime-q_0|}{b}}$, and $q^\prime > q$ if $q^\prime_0 < q_0$.
Hence, without loss of generality we can assume that $q^\prime \ge q$, so we are examining 
$$\max_{\substack{q, q^\prime \in D \\ 0 \le q^\prime-q \le \dq}} \frac{C_{q^\prime}}{C_q} e^{\frac{q^\prime-q}{b}}.$$

Equivalently, since $q^\prime \ge q$, we can consider $\max_{\substack{q \in D \\ 0 \le z \le \dq}} \frac{C_{q+z}}{C_q} e^{\frac{z}{b}}$.

By Lemma~\ref{lm:dcdq}, $\frac{\partial}{\partial q} \left( \frac{C_{q+z}}{C_q} e^{\frac{z}{b}}\right) \le 0$, hence the maximum is attained at the smallest possible $q$, \ie\ $$\max_{\substack{q \in D \\ 0 \le z \le \dq}} \frac{C_{q+z}}{C_q} e^{\frac{z}{b}} = \max_{0 \le z \le \dq} \frac{C_{l+z}}{C_l} e^{\frac{z}{b}}.$$

Similarly, by Lemma~\ref{lm:dcdz}, $\frac{\partial}{\partial z} \left( \frac{C_{q+z}}{C_q} e^{\frac{z}{b}}\right) \ge 0$, hence the maximum is attained at the largest possible $z$, giving $\max_{\substack{q \in D \\ 0 \le z \le \dq}} \frac{C_{q+z}}{C_q} e^{\frac{z}{b}} = \frac{C_{l+\dq}}{C_l} e^{\frac{\dq}{b}}$, as required.
\end{proof}

\subsection{Proof of Theorem~\ref{th:mainres}}\label{sc:a:th:mainres}

\begin{proof}[Proof (Theorem~\ref{th:mainres})]
We follow a similar method of proof as used in Example~5 of~\cite{HLM15}.

Given $A \subseteq D$, and noting that $\Prob(W_q \in A) = \frac{1}{C_q} \Prob(Y_q \in A)$, where $Y_q$ is given by (\ref{eq:laplacemech}), we are seeking to show that
$$\frac{1}{C_q} \Prob(Y_q \in A) \le e^\epsilon \frac{1}{C_{q^\prime}} \Prob(Y_{q^\prime} \in A) + \delta,$$
for any measurable $A \subseteq D$ and where $q, q^\prime \in D$ and $|q-q^\prime| \le \dq$.
Given that $\Prob(Y_q \in A) = \int_A \frac{e^{-\frac{|x-q|}{b}}}{2b} dx$, we have,
$$\frac{1}{C_q} \int_A \frac{e^{-\frac{|x-q|}{b}}}{2b} dx \le e^\epsilon \frac{1}{C_{q^\prime}} \int_A \frac{e^{-\frac{|x-q^\prime|}{b}}}{2b} dx + \delta.$$

Using the triangle inequality, we see that $|x-q^\prime| \le |x-q| + |q^\prime-q|$, so it is sufficient to show that \linebreak $\frac{1}{C_q} \int_A \frac{e^{-\frac{|x-q|}{b}}}{2b} dx \le e^{\epsilon-\frac{|q^\prime-q|}{b}} \frac{1}{C_{q^\prime}} \int_A \frac{e^{-\frac{|x-q|}{b}}}{2b} dx + \delta$,
or equivalently,
$$1 \le e^{\epsilon-\frac{|q-q^\prime|}{b}} \frac{C_q}{C_{q^\prime}} + \frac{C_q}{\int_A \frac{e^{-\frac{|x-q|}{b}}}{2b} dx} \delta.$$

Since $A \subseteq D$ and given the definition of $C_q$ in Definition~\ref{df:boundedlaplace}, it follows that $C_q \ge \int_A \frac{e^{-\frac{|x-q|}{b}}}{2b} dx$, hence it is sufficient to show that $1 \le e^{\epsilon - \frac{|q^\prime-q|}{b}} \frac{C_d}{C_{q^\prime}} + \delta$.

By Lemma~\ref{lm:min}, $\Delta C(b) \, e^{\frac{\dq}{b}} \ge \frac{C_{q^\prime}}{C_q} e^{\frac{|q^\prime-q|}{b}}$ when $|q^\prime-q| \le \dq$, or equivalently $\frac{1}{\Delta C(b)} e^{-\frac{\dq}{b}} \le \frac{C_q}{C_{q^\prime}} e^{-\frac{|q^\prime-q|}{b}}$, so it is sufficient to show that
\begin{equation}\label{eq:bounded:solveforb}
1 \le \frac{1}{\Delta C(b)} e^{\epsilon-\frac{\dq}{b}} + \delta.
\end{equation}
Solving (\ref{eq:bounded:solveforb}) implicitly for $b$ completes the proof.
\end{proof}

\subsection{Proof of Lemma~\ref{lm:posinitpoint}}\label{sc:a:lm:posinitpoint}

\begin{proof}[Proof (Lemma~\ref{lm:posinitpoint})]
We first note that $f(b) > 0$ if and only if $\epsilon - \log \Delta C(b) - \log(1-\delta) > 0$, or equivalently, if $\Delta C(b) < \frac{e^\epsilon}{1-\delta}$.
We assume that $\frac{e^\epsilon}{1-\delta} > 1$ (\ie\ that $\epsilon$ and $\delta$ are not simultaneously zero).

Given $b_0=\frac{\dq}{\epsilon - \log(1-\delta)}$, we see that
\begin{align}
\Delta C(b_0) &= \frac{2 - e^{-\epsilon + \log(1-\delta)} - e^{-\left(\frac{u-l}{\dq} - 1\right) \left( \epsilon - \log(1-\delta)\right)}}{1 - e^{-\frac{u-l}{\dq}\left(\epsilon - \log(1-\delta)\right)}}\nonumber\\
&= \frac{2 - \frac{1-\delta}{e^\epsilon} -\left(\frac{e^\epsilon}{1-\delta}\right)^{1-\frac{u-l}{\dq}}}{1 - \left(\frac{e^\epsilon}{1-\delta}\right)^{-\frac{u-l}{\dq}}}\nonumber\\
&= \frac{2 \left(\frac{e^\epsilon}{1-\delta}\right) - 1 -\left(\frac{e^\epsilon}{1-\delta}\right)^{2-\frac{u-l}{\dq}}}{\frac{e^\epsilon}{1-\delta} - \left(\frac{e^\epsilon}{1-\delta}\right)^{1-\frac{u-l}{\dq}}}.\label{eq:longdcq}
\end{align}

For simplicity, we relabel (\ref{eq:longdcq}) by setting $\alpha = \frac{e^\epsilon}{1-\delta}$ and $\beta = \frac{u-l}{\dq}$, giving
$$\Delta C(b_0) = \frac{2\alpha - 1 - \alpha^{2-\beta}}{\alpha - \alpha^{1-\beta}}.$$
We note that $\alpha > 1$ and $\beta \ge 1$.

Since $\max \left(2\alpha-\alpha^2\right) = 1$ and the maximum occurs at $\alpha=1$, it follows that $2\alpha-\alpha^2 < 1$ when $\alpha > 1$.
We can then make the following series of deductions:
\begin{align*}
2\alpha-\alpha^2 &< 1,\\
2\alpha - 1 &< \alpha^2,\\
2\alpha -  1 - \alpha^{2-\beta} &< \alpha^2  - \alpha^{2-\beta},\\
\frac{2\alpha - 1 - \alpha^{2-\beta}}{\alpha-\alpha^{1-\beta}} &< \alpha.
\end{align*}
Hence,
$$\Delta C(b_0) < \alpha = \frac{e^\epsilon}{1-\delta},$$
and it follows that $f(b_0) > 0$.

We can also show that $\Delta C(b_0) \ge 1$ through the following series of deductions:
\begin{align}
\alpha^{1-\beta} &\le 1,\label{eq:lm:posinitpoint}\\
\alpha^{1-\beta}(\alpha-1) & \le \alpha-1,\nonumber\\
0 \le \alpha - \alpha^{1-\beta} &\le 2 \alpha - 1 - \alpha^{2-\beta},\nonumber\\
\frac{2\alpha - 1 - \alpha^{2-\beta}}{\alpha - \alpha^{1-\beta}} &\ge 1.\nonumber
\end{align}
Hence, $\log \Delta C(b_0) \ge 0$.
It then follows that
$$\frac{\dq}{\epsilon - \log \Delta C(b_0) - \log(1-\delta)} \ge \frac{\dq}{\epsilon - \log(1-\delta)},$$
and that $f(b_0) \ge b_0$.
Furthermore, from (\ref{eq:lm:posinitpoint}), $f(b_0)=b_0$ if and only if $\dq = u-l$.
\end{proof}

\subsection{Proof of Lemma~\ref{lm:fprime}}\label{sc:a:lm:fprime}

\begin{proof}[Proof (Lemma~\ref{lm:fprime})]
From (\ref{eq:fb}), we have
$$f^\prime(b) = \frac{f(b)^2}{\dq \Delta C(b)} \frac{\partial \Delta C(b)}{\partial b},$$
hence $f^\prime \le 0$ if and only if $\frac{\partial \Delta C(b)}{\partial b} \le 0$.
From the definition of $\Delta C(b)$, after some simplification we have
\begin{align}
\frac{\partial \Delta C(b)}{\partial b} &=-\left(\frac{1}{2b\,C_l(b)}\right)^2 \bigg(\dq \left(e^{-\frac{\dq}{b}} + e^{-\frac{2(u-l)-\dq}{b}}\right)\nonumber\\
&\qquad\qquad + e^{-\frac{u-l}{b}} (u-l-\dq) \left(e^{\frac{\dq}{b}} + e^{-\frac{\dq}{b}}\right)\nonumber\\
&\qquad\qquad - 2(u-l)e^{-\frac{u-l}{b}}\bigg)\nonumber\\
&\le - \left(\frac{1}{2b\,C_l(b)}\right)^2 \bigg(\dq \left(e^{-\frac{\dq}{b}} + e^{-\frac{2(u-l)-\dq}{b}}\right) - 2 \dq e^{-\frac{u-l}{b}}\bigg)\label{eq:lm7prf}\\
&= - \left(\frac{1}{2b\,C_l(b)}\right)^2 \dq e^{-\frac{\dq}{b}} \left(1-e^{-\frac{u-l-\dq}{b}}\right)^2\nonumber\\
& \le 0,\nonumber
\end{align}
where (\ref{eq:lm7prf}) follows since $e^a + e^{-a} \ge 2$ for all $a \in \R$.
Note that we have $\frac{\partial \Delta C(b)}{\partial b}=0$ if and only if $\dq = u-l$.
Also note that this result holds for all $b \neq 0$, and therefore for all $b \ge b_0$.

We therefore conclude that $f^\prime(b) \le 0$ for all $b \ge b_0$, and furthermore that $f^\prime(b) = 0$ if and only if $\dq = u-l$.
\end{proof}

\end{document}